\documentclass[3p]{elsarticle}
\usepackage{amsmath,amssymb,amsfonts}
\usepackage{graphicx}
\usepackage{float}
\graphicspath{{figures/}}
\usepackage{booktabs}
\usepackage{multirow}
\usepackage{siunitx}
\usepackage{enumitem}
\usepackage{hyperref}
\usepackage{xcolor}
 \usepackage{booktabs}
 \usepackage{multirow}
 \usepackage{amsmath}
 \usepackage{microtype}
 \usepackage{siunitx}  
 \usepackage{xcolor}   
\newtheorem{remark}{Remark}

\newtheorem{proposition}{Proposition}
\newtheorem{proof}{Proof}

\definecolor{addedorange}{RGB}{220,120,20}
\usepackage{siunitx}
\DeclareSIUnit{\degCA}{\degree CA}

\begin{document}

\begin{frontmatter}

\title{Learning-Based Decision Making for Combustion Phasing\\ Control in Multi-Fuel CI Engines with Latent Fuel Reactivity Estimation
\author[umn]{Rajasree Sarkar\fnref{equal}}
\author[umn]{Aditya Satish Patil\fnref{equal}}
\author[umn]{Arunava Banerjee\fnref{CA}}
\author[umn]{Ihsan Berk Altiner}
\author[umn]{Zongxuan Sun}
\author[arl]{Kenneth Kim}
\author[arl]{Chol-Bum Mike Keown}
\address[umn]{Department of Mechanical Engineering, University of Minnesota Twin Cities, Minneapolis, MN 55455, USA.}
\address[arl]{DEVCOM Army Research Laboratory, Aberdeen Proving Ground, MD, USA, 21005.}
\fntext[equal]{These authors contributed equally to this work.}
\fntext[CA]{Corresponding author: Arunava Banerjee, Department of Mechanical Engineering, University of Minnesota Twin Cities, Minneapolis, MN 55455, USA. email: abanerje@umn.edu}
}

\begin{abstract}

Multi-fuel compression-ignition engines offer fuel flexibility but introduce uncertain, time-varying fuel reactivity, represented by cetane number (CN), which complicates cycle-to-cycle combustion-phasing control. This work formulates CA50 regulation under latent CN variation as a partially observable sequential decision problem and systematically evaluates controllers with increasing temporal and representational capacity including LinUCB, history-augmented contextual bandits, observation-only DDPG, recurrent DDPG, and a proposed GRU-guided RL framework. A Gaussian-process surrogate trained on experimental multi-fuel engine data provides a controlled and reproducible evaluation environment. Results show that myopic and fixed-history bandit methods degrade under CN variation, observation-only RL suffers from latent-state aliasing, and generic recurrence is insufficient when CN evolves rapidly. The proposed framework learns a compact GRU-based representation of fuel reactivity from combustion history and conditions both actor and critic on this estimated signal rather than oracle CN. By training the policy on the same imperfect fuel-reactivity information available at deployment, the controller avoids train-deploy inconsistency in conventional online estimate-then-control pipelines. Across unseen CN trajectories, the policy achieves stable CA50 regulation with mean absolute tracking error below $0.25^\circ$ CA at the training setpoint, while producing smooth, physically consistent SOI and glow-plug-power actuation. These results show that combustion control under latent, continuously evolving fuel dynamics requires more than standalone estimation or generic recurrence. By aligning fuel-reactivity inference with control policy learning, the proposed framework enables reactivity-aware decision-making using the same estimated state available during deployment.
\end{abstract}


\end{frontmatter}

\section{Introduction}

The transition toward low-carbon propulsion in aviation, unmanned aerial systems, and distributed energy has intensified interest in multi-fuel compression-ignition (CI) engines' capability of operating across a wide spectrum of sustainable fuel blends \cite{pal2024data, splitter2014fuel}. A central challenge in such systems is maintaining consistent combustion behavior despite substantial variability in fuel properties \cite{dong2022data}. Combustion phasing typically quantified by the crank angle corresponding to 50\%  fuel energy release (CA50) is a primary control objective, as it directly governs thermal efficiency, emissions, and engine operability \cite{halbe2017control}. In practice, CA50 must be regulated within a narrow tolerance band (on the order of a few crank-angle degrees) to ensure optimal performance while respecting hardware and safety constraints.  The fuel reactivity, commonly characterized by cetane number (CN), varies significantly across fuel batches, blending ratios, and alternative fuel types, and is often not directly measurable in real-time. Since CN strongly influences ignition delay and subsequent combustion phasing, even moderate variations can induce substantial deviations in CA50, leading to efficiency loss, increased emissions and combustion noise, and, in extreme cases, misfire or knock. Consequently, achieving reliable cycle-to-cycle regulation of CA50 under unknown and time-varying CN remains a challenging and practically critical control problem.

\subsection{Limitations of Conventional Combustion Controllers}

Production ECU lookup tables encode injection timing as a static function, embedding fuel-specific calibrations that must be fully redetermined when fuel properties change \cite{Yang2013, Yoon2007}. Closed-loop cylinder-pressure feedback only partially compensates for this mismatch during fuel switching~\cite{Willems2010}, and adaptive LUT updating~\cite{zurbriggen2014optimal} extends tolerance only to gradual, bounded drift. Gain-scheduled PID controllers face a more fundamental obstacle. Since CN modifies autoignition activation energy, a change in fuel reactivity is structurally equivalent to an operating-point shift, rendering every existing gain schedule fuel-specific by construction~\cite{killingsworth2009hcci, arora2017real}. Model predictive control, the current state of the art for RCCI and low-temperature combustion~\cite{raut2018dynamic, pamminger2022model}, intensifies rather than resolves this dependency, while physics-based MPC embeds fuel chemistry into integral model parameters, limiting robustness \cite{mishra2021design}. Data-driven MPC substitutes distributional assumptions that are also not very reliable, as reported in \cite{peng2023model}. The accuracy degradation on low-reactivity oxygenated fuels requires full reparameterisation, and as experimentally demonstrated in \cite{ansari2019experimental}, varying CN in a predefined high range produces qualitatively distinct combustion behavior that no single calibrated model can represent, confirming that CN variation constitutes a genuine out-of-distribution event for all model-based controllers. Adaptive methods \cite{larimore2015adaptive} and extremum seeking approaches \cite{killingsworth2009hcci, pla2021line} track gradual parametric drift but violate their quasi-static assumptions under rapid CN transients; extremum seeking alone requires on the order of minutes to converge per operating point \cite{killingsworth2009hcci}. Across all paradigms, the shared architectural limitation is the same that is, adaptation remains reactive, model-dependent, and too slow for the timescale of real-world CN variability. Observer-based approaches partially address this limitation by recovering fuel-related quantities from in-cylinder pressure \cite{pla2021line}, though they target proxies rather than CN directly and require sensing infrastructure, that is absent from most production platforms. More broadly, recent data-driven approaches have explored explicit fuel-state estimation as an intermediate step in the control loop. In particular, our prior studies \cite{banerjee2025data, pal2024data} demonstrated that CN can be inferred online from combustion measurements and used for combustion phasing control during discrete fuel switches in UAS engines, establishing the feasibility of data-driven fuel-state inference in closed-loop control. Extending this paradigm to continuously varying CN introduces a structural consideration in which treating CN as an explicit intermediate variable couples closed-loop performance to the accuracy and latency of the estimation stage. While this estimate-then-control architecture performs well under discrete fuel switching, its separated design becomes vulnerable under rapid CN evolution. The controller is optimized assuming an accurate CN estimate, so residual estimation bias and transient lag directly propagate into the control input without an explicit corrective mechanism. Since the controller is not trained under the estimator's approximation errors, performance degrades precisely when estimation is most difficult, especially during fast CN transients. This motivates a shift in architecture. Under continuously varying and unobserved fuel reactivity, this vulnerability is mitigated by learning a deployable latent representation offline so that the actor and critic are conditioned on the same noisy inferred fuel-reactivity signal available at deployment, rather than on oracle CN as a policy-state variable. Rather than decoupling performance from estimation fidelity entirely, this design ensures that the policy implicitly learns to be robust to the GRU's approximation errors, because the controller never receives oracle CN as an input during policy training or execution.

\subsection{Learning-Based Formulations and Challenges}

Learning-based decision-making frameworks, such as reinforcement learning (RL) and contextual bandits (CB), provide a natural mechanism to integrate decision-making directly with observed system behavior through policy learning from operational experience, without requiring an explicit plant model or a separate intermediate estimation stage. Recent work has established RL across a range of combustion control tasks, including injection timing optimization in multi-pulse CI engines \cite{henry2022deep}, multi-fuel strategy discovery \cite{wimer2023deep}, safe emission control on diesel platforms \cite{norouzi2023safe}, cycle-to-cycle variability regulation on real SI engines \cite{maldonado2024reinforcement}, and fuel-transfer on real HCCI hardware \cite{bedei2025safe}. An offline plant model has also been used to train RL policies for hydrogen-diesel dual-fuel operation with subsequent experimental validation \cite{sharma2025safe}. However, none of these works addresses real-time CA50 regulation under \emph{unknown, continuously varying} CN.

The structure of the CA50 problem makes the formulation choice non-trivial. Combustion responds to injection timing through a predominantly static input–output map, with intake conditions resetting each cycle, such that cycle-to-cycle variation is not driven by action-induced state evolution. This suggests a bandit-like structure, where a contextual bandit (CB) agent selects actions to maximize immediate reward conditioned on the current operating context. However, this apparent suitability breaks down once CN variability is accounted for. When CN evolves as a latent exogenous process, current CN is temporally correlated with previous CN through fuel supply dynamics, and further, remains unobserved. As a result, the effective context is neither fully observable nor i.i.d., violating the fundamental assumptions underlying contextual bandits. History augmentation partially recovers the latent CN signal from recent combustion traces, but a fixed-length window saturates under rapid CN evolution. Standard observation-based RL policies inherit the same aliasing vulnerability through the POMDP partial-observability problem \cite{ghosh2021generalization, liu2022partially, omidshafiei2017deep}. In the absence of an internal mechanism to represent latent fuel reactivity, distinct CN states may generate identical observations, compelling the policy to produce identical control actions despite fundamentally different underlying combustion dynamics, thereby inducing systematic performance degradation. Collectively, these observations reveal that the core difficulty is not tied to a specific control paradigm, but arises from the information structure of the problem itself. Unknown and continuously evolving CN introduces a latent, temporally correlated state that is not directly observable, rendering both myopic contextual decision-making and observation-only RL insufficient. This establishes CA50 regulation under fuel variability as a fundamentally partially observable sequential decision problem, where effective control requires both recovery of latent fuel-state information from history and its direct integration into the decision-making process.

\subsection{Learning-Based Combustion Phasing Control}

To address the above mentioned challenges, this paper investigates five learning formulations of increasing temporal and representational capacity: (i) LinCB, (ii) history-augmented CB (H-CB), (iii) observation-only DDPG, (iv) recurrent DDPG (RDPG), and (v) the proposed GRU-RL framework. All methods are evaluated on a Gaussian process surrogate trained on experimental engine data, which provides a controlled and reproducible environment for comparing decision-making architectures under identical CN-variation scenarios. The progression is designed to answer a specific question: when CN is unobserved and evolves continuously, how much can be recovered by immediate contextual decision making, how much by fixed-window history augmentation, and when is an explicitly sequential policy required? LinCB tests the bandit assumption directly. H-CB examines whether short history is sufficient to recover the missing fuel-state information. Observation-only DDPG tests whether replacing CB with RL alone resolves the partial-observability limitation. RDPG then introduces generic recurrence to capture latent temporal structure. Finally, the proposed GRU-RL framework builds on our earlier CN-estimation-based studies by using a GRU trained offline to extract a deployable fuel-state estimate from recent combustion observations, which is then used by the actor--critic policy for CA50 regulation without access to oracle CN during policy learning or deployment. The comparison shows that fixed-window history augmentation improves over standard CB but degrades as CN variation becomes faster and more severe, while observation-only RL remains vulnerable to the same latent-state ambiguity. Recurrent policies perform substantially better, confirming that temporal inference is essential in this problem. Among them, the proposed GRU-RL framework delivers the most reliable CA50 regulation across unseen CN trajectories, indicating that supervised recovery of fuel-state information from history, when coupled to sequential policy optimization, is more effective than either myopic contextual decision making or recurrence alone. The contributions of this paper are summarized as follows:

\begin{enumerate}[leftmargin=*, label=(\arabic*)]

\item \textit{Systematic CB-to-RL progression:}
A unified evaluation of LinCB, H-CB, observation-only DDPG, RDPG, and the proposed GRU-RL framework (CN-augmented DDPG) is performed on a common GP-based surrogate isolating the role of temporal information and model capacity. This revealed that for the considered engine problem, CB is fundamentally limited by latent-state ambiguity, while fixed-window history provides only partial mitigation. On the other hand, observation-only RL does not resolve partial observability, thereby motivating the proposed GRU-based approach.

\item \textit{Proposed GRU-guided sequential RL framework for fuel-flexible CA50 regulation:}
A controller is developed in which a GRU, trained offline using privileged CN information, provides a deployable fuel-state estimate from recent combustion observations, and an actor--critic policy uses this estimate for closed-loop CA50 regulation without oracle CN at training or deployment.

\item \textit{Demonstrated improved robustness to latent CN variation via integrating fuel-state inference with sequential decision making:}
The proposed GRU-RL framework demonstrated the most reliable regulation across increasingly severe and unseen CN trajectories, consistently outperforming both myopic and history-augmen-ted baselines as well as generic recurrent policies.
\end{enumerate}

The rest of the paper is organized as follows: Section~\ref{sec:problem} states the control problem along with introducing the operating scenarios. Section~\ref{contextBand} and Section~\ref{sec:methods} details the CB and RL methodologies (baselines and proposed), including reward design and CN-augmented architectures, respectively. Section~\ref{sec:results} presents results and discussion, including generalization to unseen CN trajectories and inference-time measurements. Section~\ref{sec:conclusion} concludes and outlines future directions.

\section{Problem Formulation}\label{sec:problem}
\subsection{Engine Variables and Control Objective}

We consider a single-cylinder multi-fuel compression-ignition (CI) engine operated at fixed speed and load, where combustion phasing is regulated on a cycle-by-cycle basis. Let $k \in \mathbb{N}$ denote the engine cycle index. At each cycle, the controller selects the control input vector 
\begin{equation}
u_k \triangleq 
\begin{bmatrix}
\mathrm{SOI}_k \\
\mathrm{GPP}_k
\end{bmatrix},
\end{equation}
where $\mathrm{SOI}_k$ denotes the start-of-injection timing ($^\circ$ CA) and $\mathrm{GPP}_k$ denotes the glow plug power (W), which provides ignition assistance. The controlled output $y$ is the crank angle corresponding to $50\%$ cumulative heat release (denoted as CA50),
\begin{equation}
y_k \triangleq \mathrm{CA50}_k,
\end{equation}
which serves as the primary indicator of combustion phasing and reflects the combined influence of injection timing, ignition assistance, and fuel reactivity. Fuel reactivity is characterized by the cetane number (CN), modeled as an operating parameter
\begin{equation}
\theta_k \triangleq \mathrm{CN}_k,
\end{equation}
which may vary over time due to fuel blending and switching. Variations in $\theta_k$ significantly influence ignition delay and thereby induce shifts in combustion phasing. In practice, $\theta_k$ is difficult to be directly measured and hence, could be inferred indirectly from observed combustion responses.

The control objective is to regulate CA50 to a desired reference $y^\star$ by suitably choosing $u_k$ at each cycle while meeting actuator and safety constraints. Particularly, the control inputs must lie within admissible operating bounds
\begin{align}
\mathrm{SOI}_{\min} \le \mathrm{SOI}_k \le \mathrm{SOI}_{\max}, \\
\mathrm{GPP}_{\min} \le \mathrm{GPP}_k \le \mathrm{GPP}_{\max},
\end{align}
and their cycle-to-cycle variation is limited by actuator rate constraints
\begin{align}
|\mathrm{SOI}_k - \mathrm{SOI}_{k-1}| \le \Delta_{\mathrm{SOI}}, \\
|\mathrm{GPP}_k - \mathrm{GPP}_{k-1}| \le \Delta_{\mathrm{GPP}},
\end{align}
which prevent abrupt changes in injection timing or glow plug power between successive cycles.

Defining the tracking error
\begin{equation}
e_k \triangleq y^\star - y_k,
\end{equation}
the control problem is to design a policy $\pi$ that maps available measurements and past observations to the control input $u_k$ such that $e_k \to 0$ as $k \to \infty$. It is to be noted that CA50 is governed by strongly nonlinear combustion physics and must be controlled through two coupled inputs, SOI and GPP, despite only a single measured output. On the other hand, GPP affects combustion indirectly through ignition assistance and is subject to energy constraints. Excessive GPP, particularly with high-CN fuels, can over-advance autoignition, intensify heat-release rates, and drive the maximum pressure-rise rate (MPRR) beyond safe limits, increasing the risk of knock-like abnormal combustion and potentially severe mechanical damage. The problem is further complicated by the fact that fuel reactivity, parameterized by CN is time-varying and not directly measured online. Thus, the problem can be best viewed as a constrained, nonlinear, partially observed combustion-phasing control problem rather than a simple output regulation problem.

\subsection{Central challenges of hidden CN}
\label{sec:info_structure}

For unknown CN, the relevant execution-time latent state is not fully observed. A convenient minimal representation can be
\begin{equation}
\xi_k \triangleq (e_k,\Delta e_k,\theta_k),
\end{equation}
where $\theta_k$ is the unmeasured cetane number and $\Delta e_k=e_k - e_{k-1}$. The controller does not observe $\xi_k$ directly; instead, it receives the measurable combustion quantities
\begin{equation}
o_k \triangleq (e_k,\Delta e_k),
\end{equation}
and may additionally maintain a history
\begin{equation}
h_k \triangleq (o_1,u_1,\ldots,o_{k-1},u_{k-1},o_k)
\end{equation}
from which information about $\theta_k$ can only be inferred indirectly. The hidden-CN problem is therefore a partially observable control problem where different latent fuel states can induce the same instantaneous observable error coordinates while still requiring different SOI/GPP actions to remain near the CA50 target. The following two propositions formalize the central difficulty of hidden-CN combustion control. 

\begin{proposition}
\label{prop:reactive_insufficiency}
\leavevmode\\\textnormal{\textbf{Reactive observation insufficiency.}}
Suppose there exist two latent states $\xi_k^{(1)}=(e,\Delta e,\theta^{(1)})$ and $\xi_k^{(2)}=(e,\Delta e,\theta^{(2)})$ with $\theta^{(1)} \neq \theta^{(2)}$ that induce the same observable state $o_k=(e,\Delta e)$, but whose optimal actions satisfy $u^\star(\xi_k^{(1)}) \neq u^\star(\xi_k^{(2)})$. Then no deterministic reactive policy of the form $\pi(o_k)$ can be optimal for both latent states.
\end{proposition}

\begin{proof}
A deterministic reactive (i.e. memoryless) policy maps the common observation $o_k=(e,\Delta e)$ to a single action $\pi(o_k)$. If the two latent states require different optimal actions, then $\pi(o_k)$ can match at most one of them. Therefore the same reactive policy must be suboptimal for at least one latent state. \qed
\end{proof}

\begin{proposition}
\label{prop:cb_mixture}
\leavevmode\\\textnormal{\textbf{Observable-context reward mixing under hidden CN.}}
Let $x_k$ denote the observable context available to a contextual bandit under hidden CN. Then the conditional reward law satisfies
\begin{equation}
\mathbb{E}[r_k \mid x_k]
=
\sum_{\theta_k} p(\theta_k \mid x_k)\,\mathbb{E}[r_k \mid x_k,\theta_k],
\end{equation}
so $x_k$ is not a sufficient statistic for the reward distribution when CN is unobservable and temporally correlated.
\end{proposition}

\begin{proof}
If CN is hidden, the agent cannot condition on the full context $(x_k,\theta_k)$ and instead observes only the mixture indexed by $x_k$. Because $\theta_k$ evolves through the unobserved fuel process, the mixture weights $p(\theta_k \mid x_k)$ inherit dependence on latent history rather than on the observable context alone. Hence the i.i.d.\ contextual-bandit assumption fails for the \emph{observable} context seen by the agent, even though it would hold for the hypothetical full context that includes CN. \qed
\end{proof}

Proposition~\ref{prop:reactive_insufficiency} explains why observation-only policies can look acceptable in CA50 space while still converging to physically wrong actuation. Proposition~\ref{prop:cb_mixture} explains why contextual-bandit reward models become mixtures over latent fuel states unless some form of history-based inference or latent-state estimation is added. The sections that follow evaluate precisely these different responses to the same information-structure problem.

\section{Contextual Bandit Formulation} \label{contextBand}
\subsection{Bandit Structure of the Combustion Control Problem}
At fixed speed and load, combustion responds to injection timing and glow plug power through a predominantly static input-output map where the combustion outcome $y_k$ is governed primarily by the current control input $u_k$ and prevailing fuel reactivity $\theta_k$,
\begin{equation}
     y_k = f(u_k,\, \theta_k) + \varepsilon_k,
 \end{equation}
 where $\varepsilon_k$ captures cycle-to-cycle combustion variability. Crucially, the system exhibits quasi-static behavior across cycles, such that cycle-to-cycle variation arises from exogenous CN fluctuations rather than from dynamical consequences of prior actions. This absence of action-induced state carry-over makes the problem structurally closer to a bandit. In the CB formulation, each combustion cycle corresponds to a decision round. At cycle $k$, the controller observes a context $x_k$, selects a control action $u_k$ from the admissible set $\bar{\mathcal{U}}$, and receives an immediate scalar reward $r_k$ encoding the CA50 tracking objective along with regularization terms for actuator smoothness and GPP usage. The CB agent does not maintain a value function over future cycles and does not propagate credit across time. 

\subsection{Admissible Action Set}

At each cycle $k$, candidate actions are evaluated over a discrete grid $\bar{\mathcal{U}} \subset \mathbb{R}^2$ enforcing actuator bounds and cycle-to-cycle rate limits as
\begin{equation}
    \bar{\mathcal{U}} = \Big\{ u = (\mathrm{SOI}, \mathrm{GPP})^\top :
    \begin{aligned}
        &\mathrm{SOI}_{\min} \le \mathrm{SOI} \le \mathrm{SOI}_{\max},\\
        &\mathrm{GPP}_{\min} \le \mathrm{GPP} \le \mathrm{GPP}_{\max},\\
        &|\mathrm{GPP} - \mathrm{GPP}_{k-1}| \le \Delta_{\mathrm{GPP}}\\
        &\mathrm{GPP} \le \mathrm{GPP}_k^{\mathrm{ceil}}
    \end{aligned}
    \Big\},
\end{equation}
 where $\mathrm{GPP}_k^{\mathrm{ceil}} = \max(\mathrm{GPP}_{k-1} - \Delta_{\mathrm{GPP}},\;\mathrm{GPP}_{\max}(1 - c_k))$ is a hard CN-dependent ceiling and \begin{align}\label{CN_est}
c_k =
\frac{\theta_k - \theta_{\min}}
{\theta_{\max} - \theta_{\min}}
\end{align}
is the normalized known measured CN. All quantities determining $\bar{\mathcal{U}}$ are observable before the action is selected and are used only to restrict the feasible set without modifying reward. The action space $\mathcal{U}$ is constructed as a Cartesian product of $n_{\mathrm{SOI}}$ uniformly spaced SOI values and $n_{\mathrm{GPP}}$ uniformly spaced GPP values, yielding $|\mathcal{U}| = n_{\mathrm{SOI}} \times n_{\mathrm{GPP}}$ discrete actions. The UCB argmax is evaluated by direct enumeration over this grid, which is computationally trivial for the grid sizes considered. Actions violating rate constraints relative to the previous cycle are excluded
from consideration during action selection.

\subsection{LinUCB with Known CN Context}
\label{sec:linucb_known}

When the cetane number $\theta_k$ is observable, the LinUCB algorithm \cite{li2010contextual, chu2011contextual} is applied with disjoint linear models, where each discrete action $u \in \mathcal{U}$ maintains its own parameter vector. The context $x_k$ at cycle $k$ is defined as
\begin{equation}
    x_k = \begin{bmatrix}
        1,\;
        c_k,\;
        c_k^2,\;
        \beta_k,\;
        e_{k-1},\;
        |e_{k-1}|
    \end{bmatrix}^{\top} \in \mathbb{R}^6,
\end{equation}
where $c_k = (\theta_k - \theta_{\min}) / (\theta_{\max} - \theta_{\min})$ denotes the normalised CN. Let $E_k = \sum_{i=1}^{k}\mathrm{GPP}_i \cdot \Delta t_c$ denote the cumulative glow plug energy consumed up to cycle $k$, where $\Delta t_c = 120/N_{\mathrm{RPM}}$ is the cycle duration. The normalised remaining GPP energy budget is then defined as $\beta_k = \max(0,\,1 - E_k/\mathcal{B}) \in [0,1]$ where $\mathcal{B}$ is prefixed GPP energy allocation. The previous-cycle tracking error is defined as $e_{k-1} = \mathrm{clip}((y^\star - y_{k-1})/\gamma_e,\,-1,\,1)$, where $\gamma_e = 10$ is set for normalization. The quadratic CN term $c_k^2$ captures the nonlinear dependence of the reward on fuel reactivity, while the absolute value $|e_{k-1}|$ enables the linear model to distinguish symmetrically between under-advancing and over-advancing conditions.

For each action $u$, the algorithm maintains $A_u \in \mathbb{R}^{d \times
d}$ and $b_u \in \mathbb{R}^d$, initialised as $A_u = I_d$ and $b_u = 0$ where $d=6$.
The UCB acquisition function selects actions according to
\begin{equation}
    u_k = \arg\max_{u \in \bar{\mathcal{U}} \subset \mathcal{U}}
    \left[
        \hat{\theta}_u^\top x_k
        + \alpha \sqrt{x_k^\top A_u^{-1} x_k}
    \right],
    \label{eq:ucb_known}
\end{equation}
where $\hat{\theta}_u = A_u^{-1} b_u$ is the ridge regression estimate of the reward parameter for action $u$, and $\alpha > 0$ controls the exploration-exploitation trade-off. After observing reward $r_k$, the parameters for the selected $u_k$ are updated as
\begin{align}\label{eqn:CB_update}
\begin{matrix}
    A_{u_k} &\leftarrow A_{u_k} + x_k x_k^\top\\
    b_{u_k} &\leftarrow ~~b_{u_k} + r_k x_k
\end{matrix}
\end{align}
This disjoint model structure treats each action independently, which is suitable when the action-reward relationship varies substantially across the discrete action space. Thus, at each cycle $k$, the context $x_k$ is processed by the LinUCB policy to select an action $u_k$. The action is applied to the engine, yielding the output $y_k$, which is subsequently mapped to a reward $r_k$ for learning and decision-making. The reward $r_k$ is designed as provided in Subsection \ref{rew_design}.

While the linear reward model underlying LinUCB is a restrictive assumption given the nonlinear dependence of CA50 on the control inputs and latent fuel reactivity, LinUCB serves as the canonical baseline for contextual bandit methods, making it a meaningful reference point. Its performance degradation under nonlinear dynamics serves as the empirical motivation for the history-augmented and recurrent nonlinear methods that follow.

\subsection{History-Augmented LinUCB with Unknown CN}
\label{sec:linucb_unknown}

The apparent structural alignment between combustion and LinUCB breaks down when $\theta_k$ is unknown, because $\theta_k$ evolves as a latent exogenous process correlated across cycles and the optimal action at cycle $k$ depends on the unobserved $\theta_k$. The result is reward aliasing where identical observable contexts can correspond to materially different optimal actions under different latent CN values. History augmentation partially mitigates this by encoding recent combustion observations into the context, allowing the agent to infer $\theta_k$ implicitly. Thus, when the cetane number is unobservable, the context must encode sufficient information to implicitly infer the latent fuel reactivity from past system behavior. Following the principle of history-based context augmentation \cite{li2010contextual,bouneffouf2012contextual}, the context vector incorporates a fixed-length window of recent observations:
\begin{equation*}
    \begin{bmatrix}
1,\;
\beta_k,\;
e_{k-1},\;
|e_{k-1}|,\;
\tilde{u}_{k-1:k-L}^{(\mathrm{SOI})},\;
\tilde{u}_{k-1:k-L}^{(\mathrm{GPP})},\;
\tilde{y}_{k-1:k-L}
\end{bmatrix}^{\top}
\end{equation*}
$\in \mathbb{R}^{4+3L}$, where $L$ is the history window length, and the tilde notation denotes min-max normalization:
\begin{align}
    \tilde{u}^{(\mathrm{SOI})}_i &= \frac{\mathrm{SOI}_i - \mathrm{SOI}_{\min}}
        {\mathrm{SOI}_{\max} - \mathrm{SOI}_{\min}}, \\
    \tilde{u}^{(\mathrm{GPP})}_i &= \frac{\mathrm{GPP}_i - \mathrm{GPP}_{\min}}
        {\mathrm{GPP}_{\max} - \mathrm{GPP}_{\min}}, \\
    \tilde{y}_i &= \frac{y_i}{y_{\mathrm{scale}}},
\end{align}
with $y_{\mathrm{scale}}=20$, chosen to normalize CA50 observations to a comparable range. The normalised remaining GPP energy budget is defined as
\begin{align}
\beta_k = \max(0,\,1 - E_k/\mathcal{B}) \in [0,1].
\end{align} 
The history window provides indirect information about the latent CN through the signature of CA50 responses to past actuator commands. The admissible action set enforces the same rate limits as in the known-CN case. The CN-dependent GPP ceiling is replaced by an ignition-sufficiency
proxy:
\begin{align*}
    \hat{c}_k &= \mathrm{clip}\!\left(
        \frac{(y^\star + \delta) - \bar{y}_k^{(L)}}{2\delta},\;0,\;1
    \right),
    \quad
    \mathrm{GPP}_k^{\mathrm{ceil}}\\
    &= \max\!\left(\mathrm{GPP}_{k-1} - \Delta_{\mathrm{GPP}},\;
                  \mathrm{GPP}_{\max}(1 - \hat{c}_k)\right),
\end{align*}
where $\bar{y}_k^{(L)}$ is the windowed mean CA50 over the last $L$ cycles and $\delta = 5\,^{\circ}$CA. The proxy $\hat{c}_k \to 0$ when CA50 is persistently above the setpoint (late combustion, GPP assistance needed) and $\hat{c}_k \to 1$ when CA50 is at or below the setpoint (GPP reduction warranted). All quantities determining $\bar{\mathcal{U}}$ are observable before action selection. Same as in the known-CN case, H-CB uses the LinUCB algorithm to employ disjoint linear models ${A_u, b_u}$ where ${u \in \mathcal{U}}$. The UCB acquisition in \eqref{eq:ucb_known} is then evaluated using the augmented context vector.

\subsection{Physics-Informed Reward Design for CB}
\label{rew_design}

The reward function $r_k$ that is to be used in \eqref{eqn:CB_update}, is constructed as a structured control objective that balances tight combustion phasing regulation, actuator constraint satisfaction, and appropriate multi-fuel ignition behavior within the learning framework. At each combustion cycle $k$, the instantaneous reward is defined as
\begin{align}
r_k = r_{\text{track}}^{(k)}  + r_{\text{GPP}}^{(k)}+ r_{\text{budget}}^{(k)},
\end{align}
where each term corresponds to a physically interpretable component of the control objective. The tracking term directly encodes the CA50 regulation objective:
\begin{equation*}
    r_{\mathrm{track}}^{(k)} =
    \begin{cases}
        100\exp(-2e_k^2), & |e_k| \le 1\,^{\circ}\mathrm{CA}, \\[4pt]
        -50\,e_k^2,       & \text{otherwise},
    \end{cases}
\end{equation*}
The CA50 response to $(u_{\mathrm{SOI}}, u_{\mathrm{GPP}}, \theta_k)$ is inherently nonlinear and no fixed $\theta^\ast$ can make $\mathbb{E}[r_k \mid x_k, a_k]$ linear in $x_k$. LinUCB is therefore applied as a heuristic under structural model misspecification as discussed in \cite{chu2011contextual,abbasi2011improved,lattimore2020learning}.

Glow plug power (GPP) usage is further shaped in a fuel-aware manner according to combustion physics and to guide learning under varying reactivity conditions. For known CN (LinUCB), the reward component $r_{\mathrm{GPP}}^{(k)}$ is designed as
\begin{equation*}
    r_{\mathrm{GPP}}^{(k)}
    = -\lambda_1\,\mathbf{1}\bigl\{\mathrm{GPP}_k > \mathrm{GPP}_{\mathrm{thr}}\bigr\}
      - \lambda_{\mathrm{mag}}\,c_k\,\tilde{u}_k^{(\mathrm{GPP})},
\end{equation*}
with threshold penalty weight $\lambda_1 = 10$, the magnitude penalty weight $\lambda_{\mathrm{mag}} = 20$, and GPP activation threshold $\mathrm{GPP}_{\mathrm{thr}} = 20\,\mathrm{W}$. The first term imposes a binary penalty whenever GPP exceeds $\mathrm{GPP}_{\mathrm{thr}}$, discouraging unnecessary ignition assist regardless of the current fuel condition. The second term is a CN-weighted GPP magnitude penalty which vanishes at $c_k = 0$ (minimum cetane number, $\theta_k = \theta_{\min}$, where maximum GPP assistance is physically warranted) and reaches its maximum magnitude $-\lambda_{\mathrm{mag}}\,\tilde{u}_k^{(\mathrm{GPP})}$ at $c_k = 1$ (maximum cetane number, $\theta_k = \theta_{\max}$, where glow plug activation is unnecessary). Together, the two terms encode the combustion physics principle that high-reactivity fuel self-ignites readily and should not incur additional ignition energy from the glow plug. On the other hand, for unknown CN (H-CB), $r_{\mathrm{GPP}}^{(k)}$ is designed as
\begin{equation*}
    r_{\mathrm{GPP}}^{(k)}
    = -\lambda_1\,\mathbf{1}\bigl\{\mathrm{GPP}_k > \mathrm{GPP}_{\mathrm{thr}}\bigr\}
      - \lambda_2\,w_k\,\tilde{u}_k^{(\mathrm{GPP})}\cdot
        \mathbf{1}\{e_k \ge 0\},
\end{equation*}
where $w_k = 0.5 + 0.5\exp(-0.5\,e_k^2) \in [0.5,\,1]$ and $\lambda_2=2$. The magnitude penalty is applied only when $e_k \ge 0$ (CA50 at or below setpoint, ignition adequate or excessive) and it is removed when $e_k < 0$ (late combustion, GPP assistance needed). Both terms depend only on current observables $(e_k, \mathrm{GPP}_k)$ and are free of history dependence and latent-variable dependence not representable in $x_k$. 

To further ensure judicious usage of GPP, a penalty based on cumulative GPP energy usage is also imposed as follows:
\begin{equation*}
    r_{\mathrm{budget}}^{(k)}
    = -\lambda_3
      \left(1 - \frac{\mathcal{B}_k}{\mathcal{B}}\right)
      \tilde{u}_k^{(\mathrm{GPP})},
\end{equation*}
where $\lambda_3=10$ and $\mathcal{B}_k = \max(0,\,\mathcal{B} - E_k)$. This term retains dependence on the full action history through $E_k$. The fixed penalty coefficient $\lambda_3$ corresponds to a fixed Lagrangian multiplier rather than an online-learned dual variable. The budget state $\beta_k$ is included in $x_k$ so that the linear model can learn a budget-contingent policy, partially compensating for the missing dual adaptation. LinUCB and H-CB are therefore positioned as near-myopic heuristic optimisers with respect to this term, consistent with established practice in applied CB engineering papers \cite{ayala2024risk,bregere2019target}. Thus, for slowly varying CN, the finite history window in H-CB is expected to retain sufficient temporal information to represent the prevailing CN signature. In contrast, when CN changes on a timescale shorter than $L$ cycles, the window necessarily averages observations across multiple CN regimes. The resulting feature representation becomes stale with respect to the current fuel reactivity, leading to degraded tracking performance.

\section{Reinforcement Learning Formulation}\label{sec:methods}

Although the static combustion map motivates CB, the unknown and time-varying CN introduces partial observability that greedy, single-step policies cannot fully resolve. From an RL perspective, the combustion control problem is a POMDP in which the latent state includes $\theta_k$, and the policy maps observation histories to control inputs,
\begin{equation}
u_k = \pi(o_{1:k}),
\end{equation}
where $o_{1:k}$ denotes the history of observed quantities up to cycle $k$. This formulation allows RL algorithms to implicitly or explicitly represent the hidden fuel state through temporal credit assignment, a capability that single-step CB policies structurally lack. This section presents the RL framework developed for combustion phasing control using the deep deterministic policy gradient (DDPG) algorithm as the continuous-control backbone. DDPG employs an actor network $\pi_\theta(s_k)$, where $s_k$ denotes the state at cycle $k$, that outputs continuous SOI/GPP actions and a critic network $Q_\phi(s_k, u_k)$ that estimates the expected discounted return. Its off-policy formulation permits efficient reuse of simulated transitions. Four state-information configurations are defined below: a known-CN oracle state, an observation-only hidden-CN state, an RDPG recurrent hidden-CN state, and the proposed GRU-estimator-augmented state. The reported unknown-CN RL comparison focuses on the three deployable hidden-CN cases.

\subsection{DDPG with Known CN (Oracle Baseline)}
For \textit{known CN}, the state representation includes both tracking information and normalized fuel reactivity,
\begin{align}
s_k =
\begin{bmatrix}
\tanh(e_k / 2) \\
\tanh(e_k - e_{k-1}) \\
\theta_k^{\text{norm}}
\end{bmatrix},
\end{align}
where $e_k = y^{\star} - CA50_k$ and $\theta_k^{\text{norm}}$ denotes the normalized CN. This state yields an approximately Markov description for value estimation. The hyperbolic tangent transformation bounds the tracking features and improves numerical conditioning during training. This case defines the fully observed oracle information structure and is used as a refer-ence state representation; the main RL results below focus on hidden-CN controllers.

\subsection{DDPG with Observation-Only State (Aliased Baseline)}
\label{sec:obs_only}

For \textit{unknown CN}, the observation-only baseline restricts the state to tracking information alone, 
\begin{align}
s_k^{\text{obs}} =
\begin{bmatrix}
\tanh(e_k / 2) \\
\tanh(e_k - e_{k-1})
\end{bmatrix},
\end{align}
thereby intentionally omitting direct information about fuel reactivity. The actor outputs normalized actions $a_k \in [-1,1]^2$, linearly mapped to the physical actuator limits of SOI and GPP. Standard DDPG training is employed with target networks, experience replay, and the deterministic policy gradient; these components form the continuous-control backbone used consistently across all RL configurations. This observation-only formulation probes a key limitation: if the state omits the latent variable governing the transition map, good CA50 tracking alone need not imply that the policy has learned the correct control mechanism.

\subsection{DDPG with Recurrent Architecture (RDPG Baseline)}
\label{sec:rdpg}

To evaluate whether temporal memory alone can compensate for the absence of direct CN measurements, a recurrent deterministic policy gradient (RDPG) baseline is implemented. In recurrent reinforcement learning, the hidden state acts as a learned summary of the observation history and can be interpreted as an implicit approximation of the belief state in a partially observable problem. In this architecture, a gated recurrent unit (GRU) processes the observation sequence,
\begin{align}
h_k = \mathrm{GRU}(h_{k-1}, s_k^{\text{obs}}),
\end{align}
where $s_k^{\text{obs}}$ consists of the normalized tracking error and its temporal difference. The hidden state $h_k$ is provided to both the actor and critic networks, allowing the policy to condition its action on accumulated historical information rather than instantaneous observations alone.

\subsection{DDPG with GRU-Based CN Estimation (GRU-RL)}\label{sec:gru_rl}
The recurrent baseline as can be seen in RDPG is important because it tests whether generic history-dependent deep RL is sufficient to overcome the state aliasing identified in Section~\ref{sec:problem}. If the observation history contains enough information about latent CN, the recurrent state should disambiguate operating conditions that are indistinguishable in instantaneous error coordinates. However, the present setting is particularly demanding where the combustion map is many-to-one, CN varies within an episode but is not directly measured, and similar short-horizon error trajectories can arise from different underlying CN values under different actuator combinations. Under such aliasing, the recurrent state may still mix trajectories associated with distinct latent fuel conditions, forcing the critic to average return targets over histories that require different optimal actions. Consequently, although recurrence increases representational capacity, it does not guarantee recovery of the latent variable that governs combustion sensitivity. This point is central to the novelty of the present study: the application exposes a concrete failure mode of standard end-to-end deep RL in safety-critical control, namely that memory alone can enrich the function approximator while still leaving the learned policy physically inconsistent or operationally fragile. Motivated by this limitation, and without modifying the underlying actor--critic learning algorithm, we introduce a CN-augmented separated estimation--control architecture that explicitly estimates fuel reactivity and incorporates it into the reinforcement learning state. Unlike online estimate-then-control pipelines where the controller is designed assuming accurate CN, here the actor–critic is optimized entirely under the GRU's noisy, lagged output. The policy therefore learns to regulate CA50 under the same approximation errors it will encounter at deployment, eliminating the mismatch that arises when a separately-calibrated estimator is coupled to a controller with different information assumptions.

\subsubsection{GRU-Based CN Estimator}

To address partial observability in the unknown-CN setting, this paper proposes a lightweight gated recurrent unit (GRU) network which is trained offline to estimate cetane number from recent input–output history. The estimator operates on a sliding window of past actuator commands and measured combustion phasing,
\begin{align}
\widehat{\theta}_k
=
f_\omega
\left(
\{\mathrm{SOI}_{k-i}, \mathrm{GPP}_{k-i}, \mathrm{CA50}_{k-i}\}_{i=0}^{n-1}
\right),
\end{align}
with a 12-cycle deployment window in all reported controller rollouts. This formulation reflects the fact that fuel reactivity is not directly measurable but influences combustion dynamics over multiple cycles. By conditioning the estimate on recent control actions and observed CA50 response, the GRU implicitly captures the temporal signature of fuel reactivity embedded in the closed-loop combustion process. The estimator parameters $\omega$ are trained using a mean squared error objective,
\begin{align}
\mathcal{L}_{\text{CN}}
=
\mathbb{E}
\left[
(\widehat{\theta}_k - \theta_k)^2
\right],
\end{align}
where $\theta_k$ denotes the true CN used in the GP-based digital twin during data generation. Training data are generated by simulating trajectories of the surrogate model under varying CN profiles and actuator inputs. To avoid learning from unreliable regions of the surrogate, samples with high GP predictive uncertainty are excluded from the training set, improving generalization across CN trajectories.

The GRU observer is the only component in the overall pipeline that is ever supervised with true CN. Once trained, the observer weights are frozen and the actor--critic loop receives only the deployable estimate $\widehat{\theta}_k$, never the oracle signal $\theta_k$ itself. The estimator therefore acts as a privileged-information distillation bottleneck: true CN is compressed offline into a low-dimensional signal that is feasible at execution time, so both actor and critic are optimized under the same noisy latent-state interface that will be available on the real engine. The relationship of this design to asymmetric critic formulations and LUPI is discussed in Remark~\ref{rem:lupi}.

The proposed architecture augments the DDPG state with the output of the shared GRU-based CN estimator described above, as
\begin{align}
s_k =
\begin{bmatrix}
\tanh(e_k / 2) \\
\tanh(e_k - e_{k-1}) \\
\widehat{\theta}_k^{\text{norm}}
\end{bmatrix},
\end{align}
where $\widehat{\theta}_k^{\text{norm}}$ denotes the normalized CN estimate. The first two components encode current tracking error and its temporal variation, while the third provides a compact representation of latent fuel reactivity inferred from recent dynamics. This augmentation transforms the problem from a partially observable setting, where identical error states may correspond to different underlying fuel conditions, into an approximately Markovian representation by explicitly reconstructing the dominant hidden variable governing combustion sensitivity. Empirically, this state augmentation stabilizes critic training and mitigates the state-aliasing effects observed in observation-only and recurrent baselines, resulting in more consistent policy convergence under time-varying CN profiles.

Conceptually, this architecture is a compact belief-state approximation rather than a generic increase in network input dimension. The GRU estimator compresses recent input--output history into an explicit proxy for the dominant latent variable, while the actor--critic module solves the continuous control problem conditioned on that proxy. This separation changes the critic's learning problem fundamentally: Bellman targets are no longer aggregated solely across trajectories that are aliased with respect to CN but are conditioned on an estimate of the hidden fuel state, reducing target inconsistency in the multivariable combustion map and enabling a control law that is both accurate and physically interpretable.

\subsubsection{Physics-Informed Reward Design for GRU-RL}

Consistent with the CB reward design in Section~\ref{rew_design}, the proposed GRU-RL controller is trained with a structured reward that balances CA50 tracking, actuator regularity, fuel-reactivity-consistent GPP usage, actuator rate feasibility, and cumulative glow-plug energy usage. The GPP terms are included for physical reasons rather than only numerical regularization: the glow plug is an electrical heating actuator with finite thermal capacity, electrical load, and duty-cycle limits. A controller that regulates CA50 by holding GPP near saturation is therefore not acceptable for deployment, even if the short-horizon tracking error is small. The cumulative GPP budget represents this finite energy and thermal-duty allocation over the 5000-cycle horizon.

At each cycle, the proposed controller receives the scalar reward
\begin{equation}
\begin{aligned}
r_{\mathrm{GRU\text{-}RL}}^{(k)}
&=
\frac{1}{100}\bigl(
r_{\mathrm{track}}^{(k)}
+r_{\mathrm{smooth}}^{(k)}
+r_{\mathrm{GPP}}^{(k)}
\\
&\quad
+r_{\mathrm{rate}}^{(k)}
+r_{\mathrm{budget}}^{(k)}
\bigr),
\end{aligned}
\end{equation}
where the factor $1/100$ is a numerical scaling used for stable actor--critic optimization. The tracking term is
\begin{align}
r_{\mathrm{track}}^{(k)}
&=
\begin{cases}
100\exp\!\left(-2\bar e_k^2\right), & |\bar e_k|\le 1^\circ\mathrm{CA},\\[3pt]
-50\bar e_k^2, & \text{otherwise},
\end{cases}
\\
\bar e_k&=\mathrm{clip}(e_k,-5,5),
\end{align}
so small tracking errors are strongly rewarded while large CA50 deviations remain heavily penalized. Smooth actuation is encouraged through
\begin{equation}
r_{\mathrm{smooth}}^{(k)}
=
-\lambda_s \left\|a_k-a_{k-1}\right\|_2,
\end{equation}
where $a_k\in[-1,1]^2$ is the normalized SOI/GPP action vector and $\lambda_s=50$.

The GPP shaping term is fuel-reactivity aware and is evaluated using the same GRU-based fuel-reactivity estimate that is supplied to the actor and critic. Thus, the reward shaping remains consistent with the deployable information interface rather than relying on oracle CN. Let
\begin{equation}
c_k=\mathrm{clip}\!\left(
\frac{\widehat{\theta}_k-\theta_{\min}}{\theta_{\max}-\theta_{\min}},
0,1
\right),
\qquad
\bar g_k=\frac{\mathrm{GPP}_k}{\mathrm{GPP}_{\max}}.
\end{equation}
For a five-cycle recent GPP window, define
\begin{equation}
\dot g_k=\frac{\mathrm{GPP}_k-\mathrm{GPP}_{k-4}}{5}.
\end{equation}
The GPP term is then
\begin{equation}
r_{\mathrm{GPP}}^{(k)}
=
r_{\mathrm{trend}}^{(k)}
-20\,c_k\bar g_k,
\end{equation}
with
\begin{equation}
r_{\mathrm{trend}}^{(k)}
=
\begin{cases}
-10\,\dot g_k(1+2c_k), & \dot g_k>0,\\[3pt]
2|\dot g_k|c_k, & \dot g_k<0 \ \mathrm{and}\ c_k>0.3,\\[3pt]
0, & \text{otherwise}.
\end{cases}
\end{equation}
This structure penalizes increasing GPP most strongly at high estimated CN, where intrinsic fuel reactivity should reduce the need for ignition assistance, and rewards reducing GPP when the estimated fuel reactivity is sufficiently high. At low estimated CN, the CN-weighted magnitude penalty weakens, allowing the controller to use more glow-plug assistance when it is physically justified.

The GPP rate-of-change term enforces actuator feasibility:
\begin{equation}
r_{\mathrm{rate}}^{(k)}
=
-\lambda_{\mathrm{rate}}
\left[
\max\left(0,
|\mathrm{GPP}_k-\mathrm{GPP}_{k-1}|-\Delta_{\mathrm{GPP}}
\right)
\right]^2,
\end{equation}
where $\Delta_{\mathrm{GPP}}=1~\mathrm{W/cycle}$ and $\lambda_{\mathrm{rate}}=50$. At the fixed engine speed used here, this corresponds to the $10~\mathrm{W/s}$ glow-plug power rate limit used in the controller implementation.

Cumulative glow-plug usage is constrained through
\begin{align}
r_{\mathrm{budget}}^{(k)}
&=
-\lambda_B
\max\left(
0,\frac{E_k-\mathcal{B}}{\mathcal{B}}
\right),
\\
E_k&=\sum_{i=1}^{k}\mathrm{GPP}_i,
\end{align}
with $\lambda_B=200$. Since engine speed is fixed in the present study, $E_k$ in W-cycle units is proportional to electrical energy consumption. The budget $\mathcal{B}$ is computed from an energy-aware ideal GPP--CN schedule with a 20\% margin,
\begin{equation}
\mathcal{B}
=
1.2\sum_{k=1}^{T}\mathrm{GPP}_{\mathrm{ideal}}(\widehat{\theta}_k),
\qquad T=5000.
\end{equation}
Thus, the controller is allowed to use substantial GPP when low fuel reactivity requires ignition assistance, but it is discouraged from treating the glow plug as an unlimited actuator. At evaluation, the same GPP rate and remaining-budget constraints are enforced on the commanded action. The observation-only DDPG and RDPG baselines were also given similar tuned hidden-CN reward variants with tracking, actuation-regularity, GPP-economy, and budget-feasibility terms, but even under tuned baseline rewards, the lack of a deployable fuel-reactivity state leads to the failure modes shown in Section~\ref{sec:results}.

\section{Results and Discussion}
\label{sec:results}

\subsection{Contextual Bandit Results}

This subsection evaluates LinUCB and history-augmented contextual bandits (H-CB) for CA50 regulation over 5000 engine cycles using the GP surrogate with additive Gaussian measurement noise (\(\sigma=0.5\)$^\circ$ CA). LinUCB is applied when CN is known, whereas H-CB is used when CN is unmeasured. The exploration parameter is set to \(\alpha=1.0\), and the H-CB history window is fixed at \(L=5\). CB simulations are conducted with $n_{\mathrm{GPP}}=71$ and $n_{\mathrm{SOI}}=11$. Quantitative results for the slowly varying and rapidly varying CN cases are summarized in Table~\ref{tab:cb_performance}. Representative responses are shown in Fig.~\ref{fig:cb_slow} for the slowly varying cases and Fig.~\ref{fig:cb_rapid} for the rapidly varying cases. 

\subsubsection{Slowly Varying CN}

Under slowly varying CN, H-CB achieves better comparatively tracking performance than LinUCB despite not using explicit CN information (Table~\ref{tab:cb_performance}). This indicates that recent input--output history provides sufficient information to infer the dominant fuel regime when CN evolves gradually. As shown in Fig.~\ref{fig:cb_slow}, LinUCB exhibits larger tracking deviations and applies higher GPP over the episode, leading to greater glow plug energy consumption. In contrast, H-CB uses recent CA50 and actuation history to adjust the control action more effectively, reducing both tracking error and GPP usage. Thus, under mild CN variability, history augmentation improves the physical efficiency of the CB policy while maintaining comparatively better CA50 regulation.

\subsubsection{Rapidly Varying CN}

\begin{figure}[htpb!]
    \centering
    \includegraphics[width=\columnwidth,height = 1.4\columnwidth]{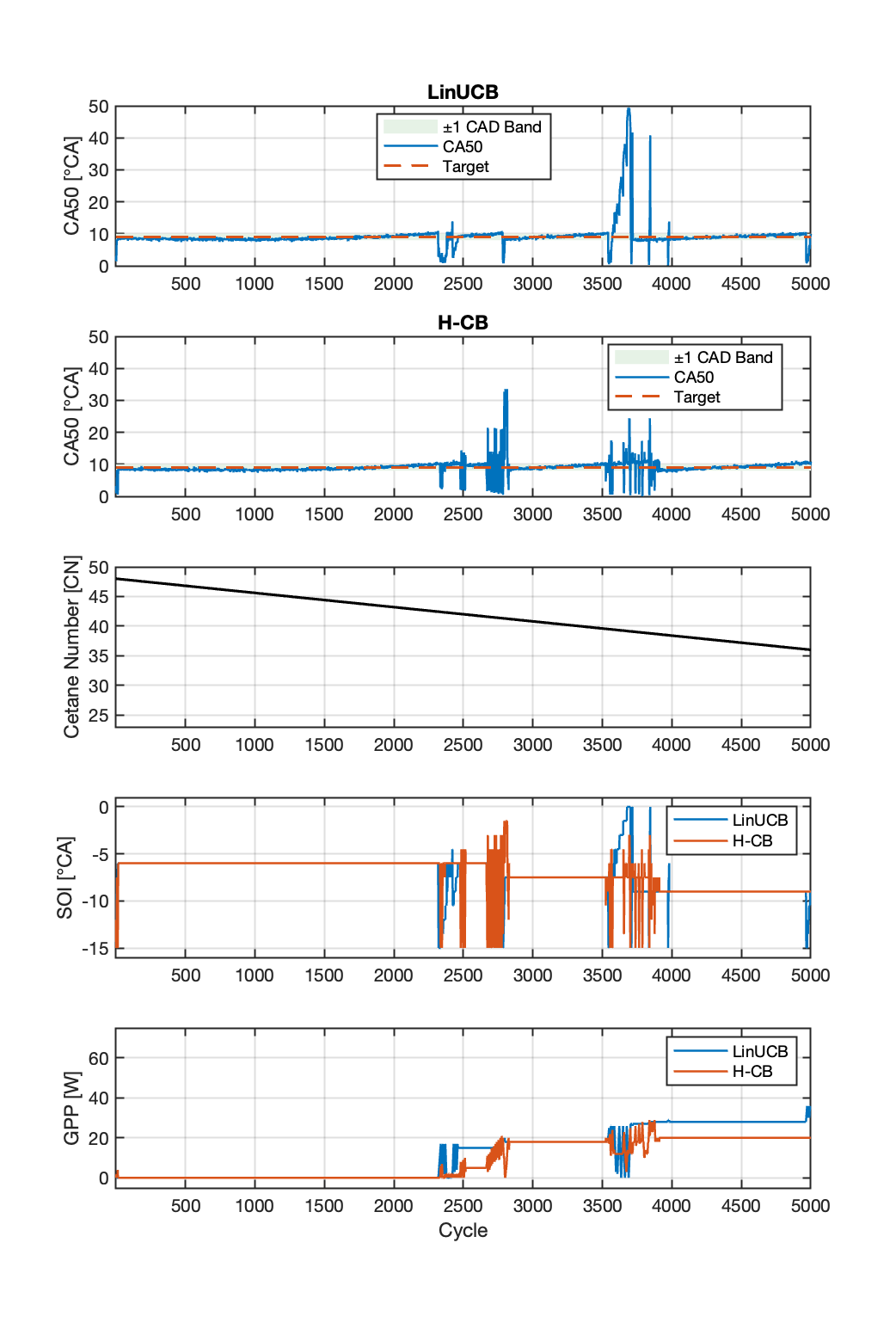}
    \caption{CB performance under slowly varying CN for known-CN LinUCB and unknown-CN H-CB.}
    \label{fig:cb_slow}
\end{figure}

Under rapid CN variation, the performance gap between the two CB formulations becomes more pronounced. LinUCB benefits from direct access to the measured CN and maintains acceptable CA50 tracking, with an RMSE of $0.99^\circ$CA (Table~\ref{tab:cb_performance}). In contrast, H-CB exhibits substantial tracking degradation, with an RMSE of $9.02^\circ$CA, because its fixed history window cannot update sufficiently fast after abrupt CN changes. This behavior is also evident in Fig.~\ref{fig:cb_rapid}, where H-CB shows large transient excursions following rapid changes in fuel reactivity. Although H-CB consumes less GPP energy than LinUCB, this reduction comes at the cost of poor combustion-phasing regulation. These results show that history augmentation is effective when CN varies slowly, but becomes insufficient when the latent fuel state changes faster than the memory window can resolve.

\begin{table}[htpb!]
\small
\centering
\caption{Performance comparison of CB methods under different CN profiles
for $\mathrm{CA50}_{\mathrm{des}}=9^\circ$\,CA, along with total glow plug
energy consumption over 5000 cycles.}
\vspace{0.1cm}
\label{tab:cb_performance}
\begin{tabular}{l l c c c}
\toprule
\textbf{Method} & \textbf{CN Profile} & \textbf{RMSE} & \textbf{MAE} & \textbf{GPP Use} \\
& \textbf{Variation} & ($^\circ$CA) & ($^\circ$CA) & (Wh) \\
\midrule
LinUCB & Known; Fixed  & 1.25 & 0.83 & 0.06 \\
LinUCB & Known; Slow   & 3.99 & 1.29 & 1.62 \\
LinUCB & Known; Rapid  & 1.63 & 0.83 & 4.88 \\
\midrule
H-CB   & Unknown; Fixed  & 0.78 & 0.69 & 0.00 \\
H-CB   & Unknown; Slow   & 2.08 & 0.99 & 1.22 \\
H-CB   & Unknown; Rapid  & 3.51 & 1.65 & 3.75 \\
\bottomrule
\end{tabular}
\end{table}

\subsubsection{Limitations of Contextual Bandits for Combustion Control}

These results establish that CB methods can partially compensate for fuel variability, but their performance is fundamentally limited by the underlying information structure and the need to satisfy physical constraints. LinUCB benefits from direct CN information, especially under rapid CN variation, but remains constrained by its reward model structure, which approximates inherently nonlinear combustion dynamics. H-CB mitigates the absence of CN through recent input-output history, but its fixed window imposes a trade-off between responsiveness and latent-state inference. Under slowly varying CN, this history is helpful to comparatively improve both tracking and GPP energy usage. However, as CN variability increases, the history window becomes stale, leading to large transients and excursions indicative of misfire-like events. More critically, both methods are myopic and lack temporal credit assignment, preventing effective use of longer-horizon structure in the latent CN process. Although H-CB can reduce GPP consumption, its rapid-CN performance shows that energy reduction alone is insufficient without reliable combustion-phasing regulation. These limitations arise from the absence of explicit latent-state reconstruction and motivate sequential RL formulations that incorporate temporal reasoning.

\begin{figure}[htpb!]
    \centering
    \includegraphics[width=\columnwidth,height = 1.4\columnwidth]{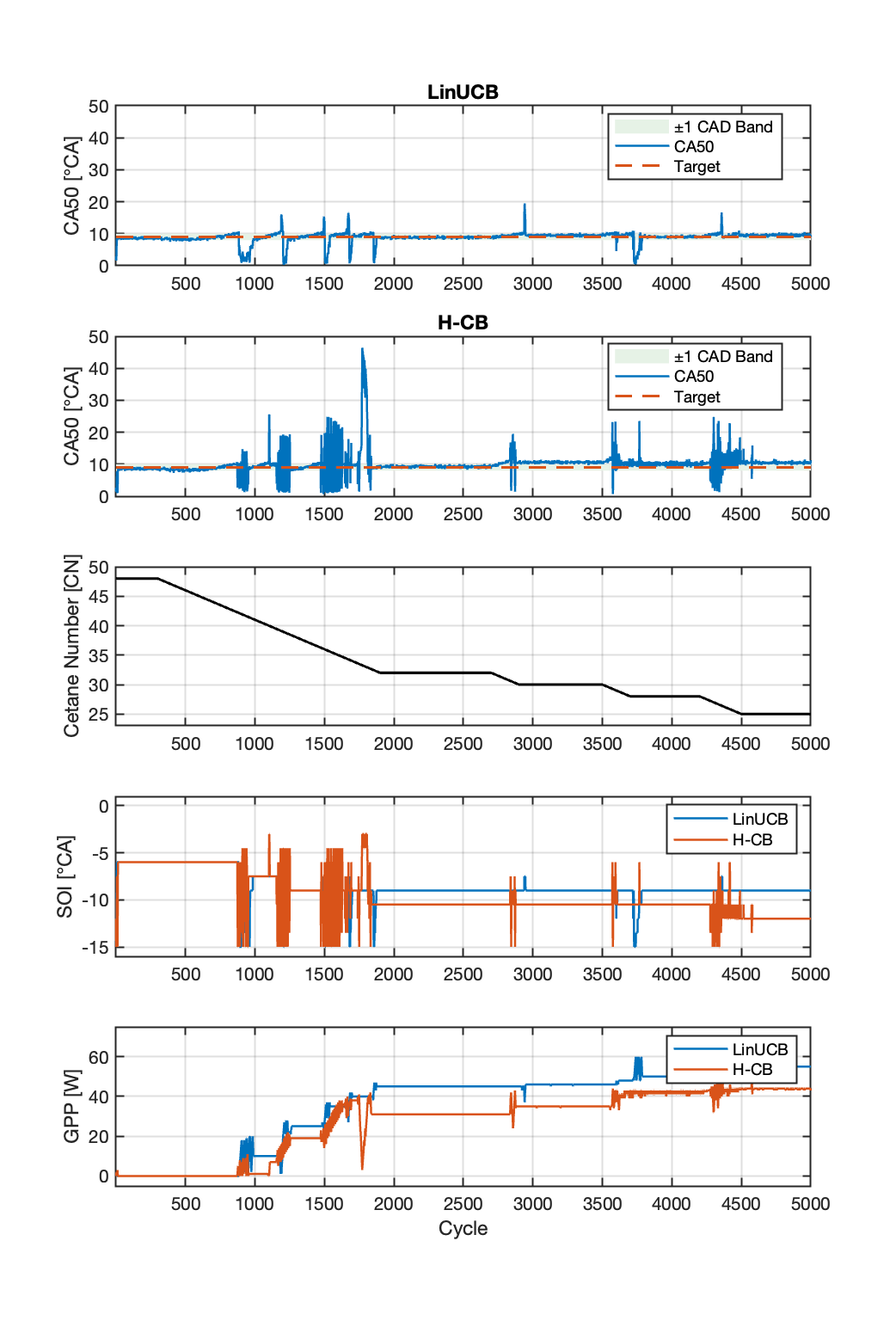}
    \caption{CB performance under rapidly varying CN for known-CN LinUCB and unknown-CN H-CB.}
    \label{fig:cb_rapid}
\end{figure}

\subsection{Reinforcement Learning Simulation Results}

All RL training and evaluation are conducted on the GP digital twin constructed from experimental combustion data. Each episode consists of 5000 combustion cycles. Policies are trained on a single CN trajectory spanning the full reactivity range and subsequently evaluated on three unseen CN profiles. A central premise of the evaluation is that CA50 tracking error alone does not adequately characterize controller quality. Because the combustion map involving $(\mathrm{SOI}, \mathrm{GPP}, \mathrm{CN})$ is many-to-one, a controller may achieve accurate phasing through a physically implausible actuation mechanism yet be unsuitable for deployment. Therefore, actuator trajectories are examined together with tracking error throughout this section to assess whether each learned policy is both accurate and physically consistent.

\subsubsection{Observation-only DDPG baseline}\label{subsec:ddpg_obs}
The observation-only DDPG agent receives only the normalized tracking error and its temporal difference, with no access to CN. Figure~\ref{fig:ddpg_obs} shows a representative rollout under a multi-rate CN profile. Although the GPP trend during first ~4500 cycles, is directionally plausible because lower fuel reactivity requires greater ignition assist, its magnitude is excessive relative to the CN-augmented baseline. The baseline was trained with GPP-economy and cumulative-budget terms, and its evaluation uses a hard GPP safety filter; nevertheless, it still fails to allocate GPP appropriately. Instead, it overuses GPP until the budget is exhausted, causing GPP to be forced to zero around cycle~4750, after which CA50 deviates abruptly and is never recovered. Such loss of combustion-phasing regulation is clearly undesirable in an engine combustion control framework. Thus, the observation-only policy does not achieve tracking comparable to the CN-augmented controller, as its apparent high-CN stability depends on an unsustainable GPP draw that is eventually terminated by the energy constraint.
\begin{figure}[htpb!]
\centering
\includegraphics[width=\columnwidth,height = 1.4\columnwidth]{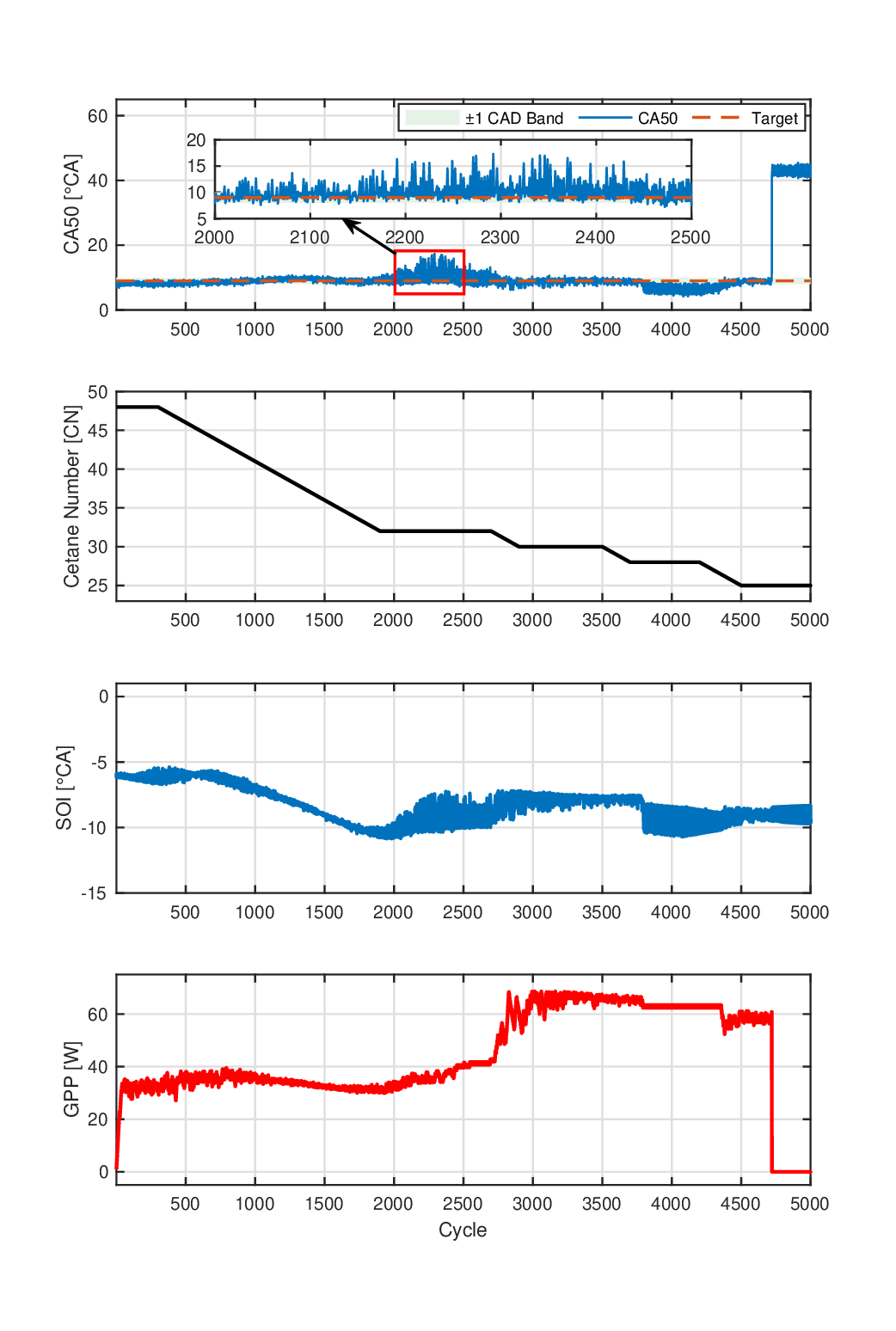}
\caption{Observation-only DDPG performance on the multi-rate CN profile with MAE $=2.80^\circ$ CA and RMSE $=8.17^\circ$ CA. Notable pathologies: (i) GPP rises to near saturation ($\sim 66~\mathrm{W}$) around cycles~3000-3500 at CN $\approx 30$, indicating excessive ignition-assist use without reactivity-state information; (ii) GPP budget exhaustion near cycle~4750 forces GPP to zero at CN $=25$, thereby diverging CA50 to $\approx 44^\circ$ CA and producing severely late combustion. The inconsistent GPP actuation shows that tracking error alone is insufficient to assess controller quality.}
\label{fig:ddpg_obs}
\end{figure}

This behavior is consistent with state aliasing under partial observability, as stated in Proposition~\ref{prop:reactive_insufficiency}. With access only to $(e_k,\Delta e_k)$, distinct CN values producing similar instantaneous errors are indistinguishable to the critic, which must fit a single value function over aliased states. The actor therefore inherits a compromise gradient that cannot associate CN level with the long-horizon cost of excessive GPP expenditure, causing near-saturation GPP in the mid-to-low-CN range instead of budget allocation proportional to true reactivity demand. This is neither a function-approximation failure nor a capacity limitation, since the low-CN, low-GPP operating point lies within the GP surrogate's training support, and the same architecture learns a stable, budget-respecting policy when augmented with the GRU-based CN estimate to be seen later. The failure is informational because $(e_k,\Delta e_k)$ is not a sufficient statistic for the optimal control law when hidden CN governs both the actuation-to-combustion mapping and the marginal cost of GPP energy. Under this information structure, the aliased critic provides no gradient signal to correct the misallocation, so the baseline fails for representational reasons rooted in the information structure rather than finite model capacity.

\subsubsection{Recurrent DDPG baseline}\label{subsec:rdpg}

The recurrent DDPG (RDPG) baseline augments the observation-only agent with a GRU-based hidden state to infer latent CN from history. Figure~\ref{fig:rdpg} shows a representative rollout under a multi-rate CN profile. Despite the added memory, $44.3\,\%$ of cycles, CA50 regulation remains outside the $\pm 1^\circ$ CA band. Two adverse intervals are highlighted by the inset zoom boxes in Fig.~\ref{fig:rdpg}. The left zoom box covers cycles~1200-2250, where in 1050-cycle window, CA50 remains persistently above the $9^\circ$ CA reference, indicating a systematic late-combustion bias rather than symmetric noise. This suggests that the recurrent state fails to adapt the control authority sufficiently fast as CN decreases through the mid-ramp region. In this window, $65.7\,\%$ of cycles violate the $\pm 1^\circ$ CA band, while SOI exhibits a cycle-to-cycle standard deviation of $1.47^\circ$ CA. Sustained retarded combustion of this magnitude can push the system toward undesirable operating regimes and violate practical combustion and hardware limits.

The right zoom box covers cycles~4000-4200 at the constant-CN plateau of CN $\approx 28$ and represents the most safety-critical interval. In this 200-cycle window of sustained oscillation, the CA50 residing in lower extreme (approx. around $3.09^\circ$ CA), indicates excessively early combustion, where heat release occurs well before top dead centre and may produce high peak in-cylinder pressures, pre-ignition, structural knock, or bearing overload. Whereas the upper extreme cases, CA50 ($\approx $ around $ 18.90^\circ$ CA) corresponds to severely late combustion, with incomplete energy conversion before the expansion stroke and increased risk of misfire, partial burn, and catalytic-converter thermal shock. The simultaneous occurrence of early and late combustion extremes within the same local neighborhood shows that the recurrent policy has not learned a stabilizing feedback law at this operating point, but instead cycles between over-advanced and over-retarded combustion phasing even though GPP remains within a physically plausible range unlike the observation-only DDPG baseline. However, RDPG remains substantially inferior to the CN-augmented controller. Its MAE is more than four times larger, and its failure mode is qualitatively different, involving persistent mid-ramp band violations and dual-sided combustion extremes rather than terminal budget collapse. These results reveal structural limitation of generic recurrent policies. 
\begin{figure}[htpb!]
\centering
\includegraphics[width=\columnwidth,height = 1.4\columnwidth]{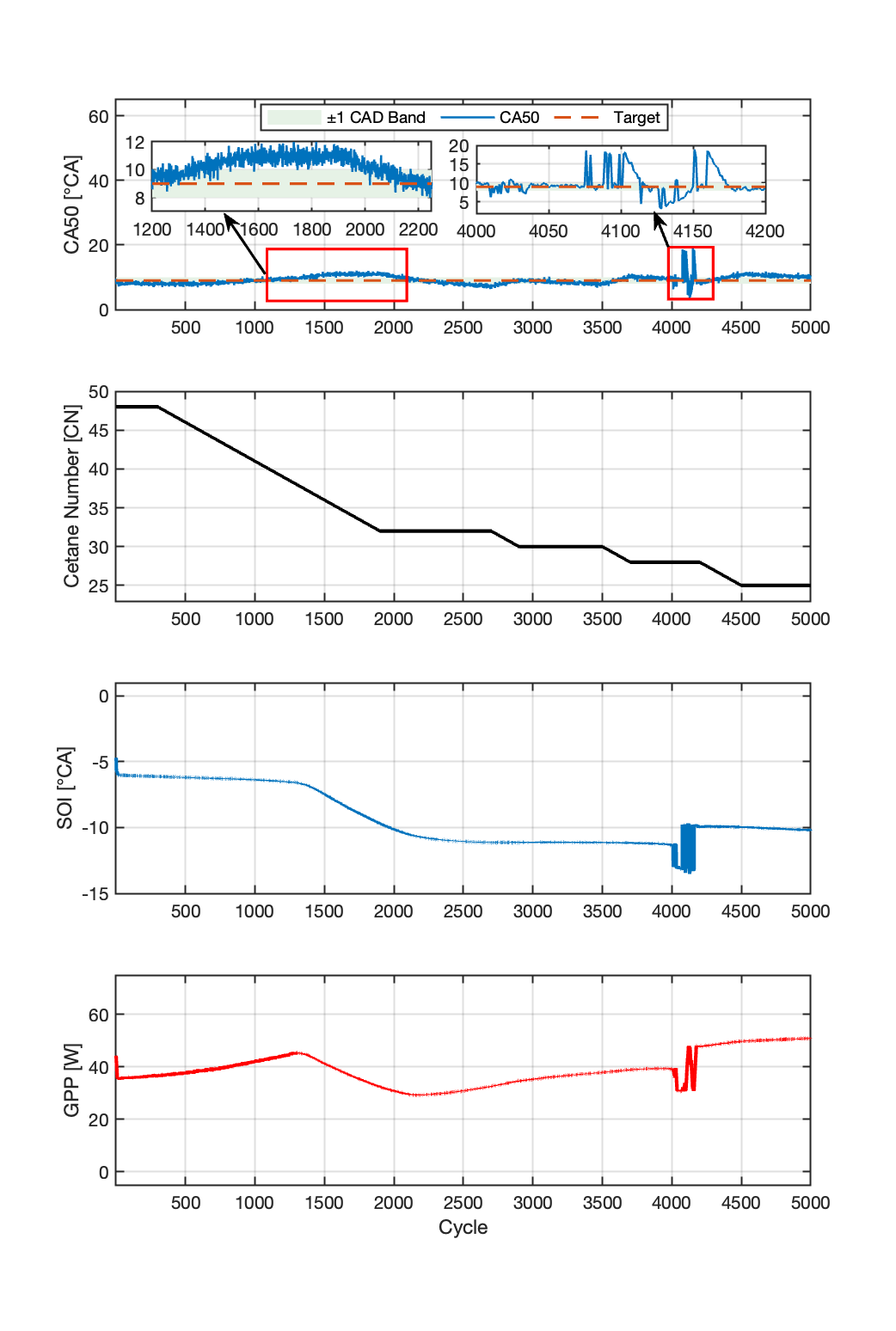}
\caption{RDPG performance on a multi-rate CN profile.
with MAE $=1.012^\circ$\,CA, RMSE $=1.294^\circ$\,CA. Notable pathologies: (i)~98\,\% band-violation rate over cycles~1501-2000 (CN~$\approx$~33); (ii) impulsive CA50 spike to $18.9^\circ$ CA at cycle~4151 (CN~$\approx$~28); (iii) 84.4\,\% band violation over the final 500 cycles (CN~=~25).}
\label{fig:rdpg}
\end{figure}

\subsubsection{CN-augmented controller}\label{subsec:cnaug}

Augmenting the RL state with the GRU-based CN estimate yields a qualitatively different and physically consistent control law. On the training CN ramp, the controller applies a strongly advanced SOI ($\approx -13.4^\circ$ CA) and moderate GPP ($\approx 40~\mathrm{W}$) at low CN ($\approx 25$), then retards SOI monotonically to $\approx -5.2^\circ$ CA and reduces GPP to $\approx 26~\mathrm{W}$ as CN increases toward $48$, matching the reduced need for ignition assist at higher fuel reactivity (Fig.~\ref{fig:cnaug_ep25}). The overall within-episode MAE against the $9^\circ$ CA target is $0.181^\circ$ CA, with the low-CN transient MAE ($\approx 0.22^\circ$ CA) only marginally higher than that in the settled high-CN region ($\approx 0.17^\circ$ CA). This behavior achieves substantially better overall tracking than observation-only DDPG and, more importantly, does so through a physically interpretable mechanism in which both actuators adapt coherently with inferred fuel reactivity rather than compensating through a spurious high-assist pattern. By supplying a compact estimate of the dominant latent variable, the proposed architecture restores an approximately Markovian state description and yields gradients consistent with the underlying combustion physics.

\begin{figure}[htpb!]
\centering
\includegraphics[width=\columnwidth,height = 1.4\columnwidth]{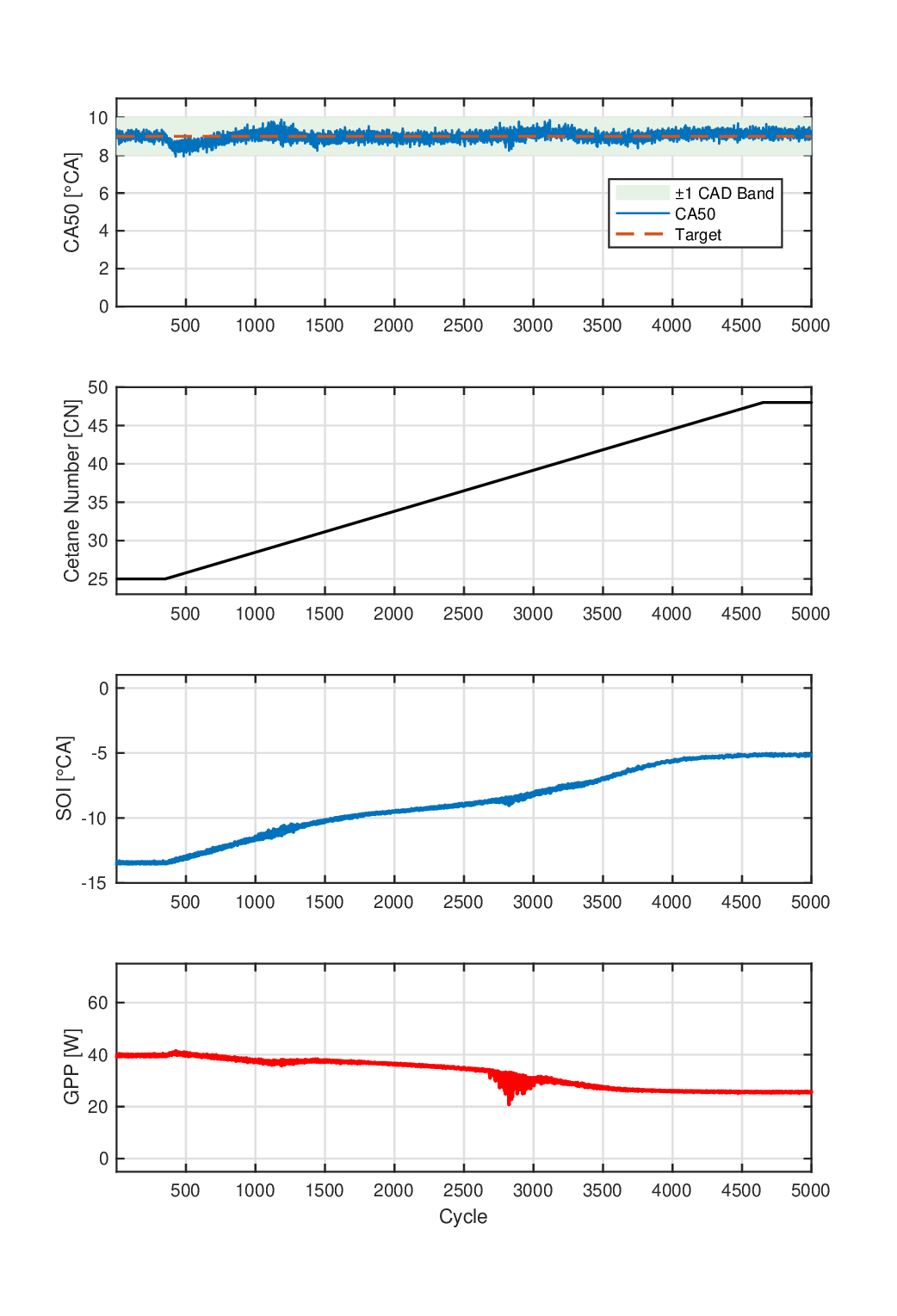}
\caption{CN-augmented DDPG performance on the training CN ramp with MAE $=0.181^\circ$ CA. As CN increases from $\approx 25$ to $48$, SOI retards from $\sim -13.4^\circ$ CA to $\sim -5.2^\circ$ CA and GPP decreases from $\sim 40$ to $\sim 26~\mathrm{W}$, demonstrating physically consistent, fuel-reactivity-aware control enabled by the GRU-based CN estimate.
}
\label{fig:cnaug_ep25}
\end{figure}

\subsubsection{Generalization to unseen CN profiles}\label{subsec:generalization}

\begin{figure*}[htpb!]
\centering
\includegraphics[width=1\textwidth,height=0.8\textwidth]{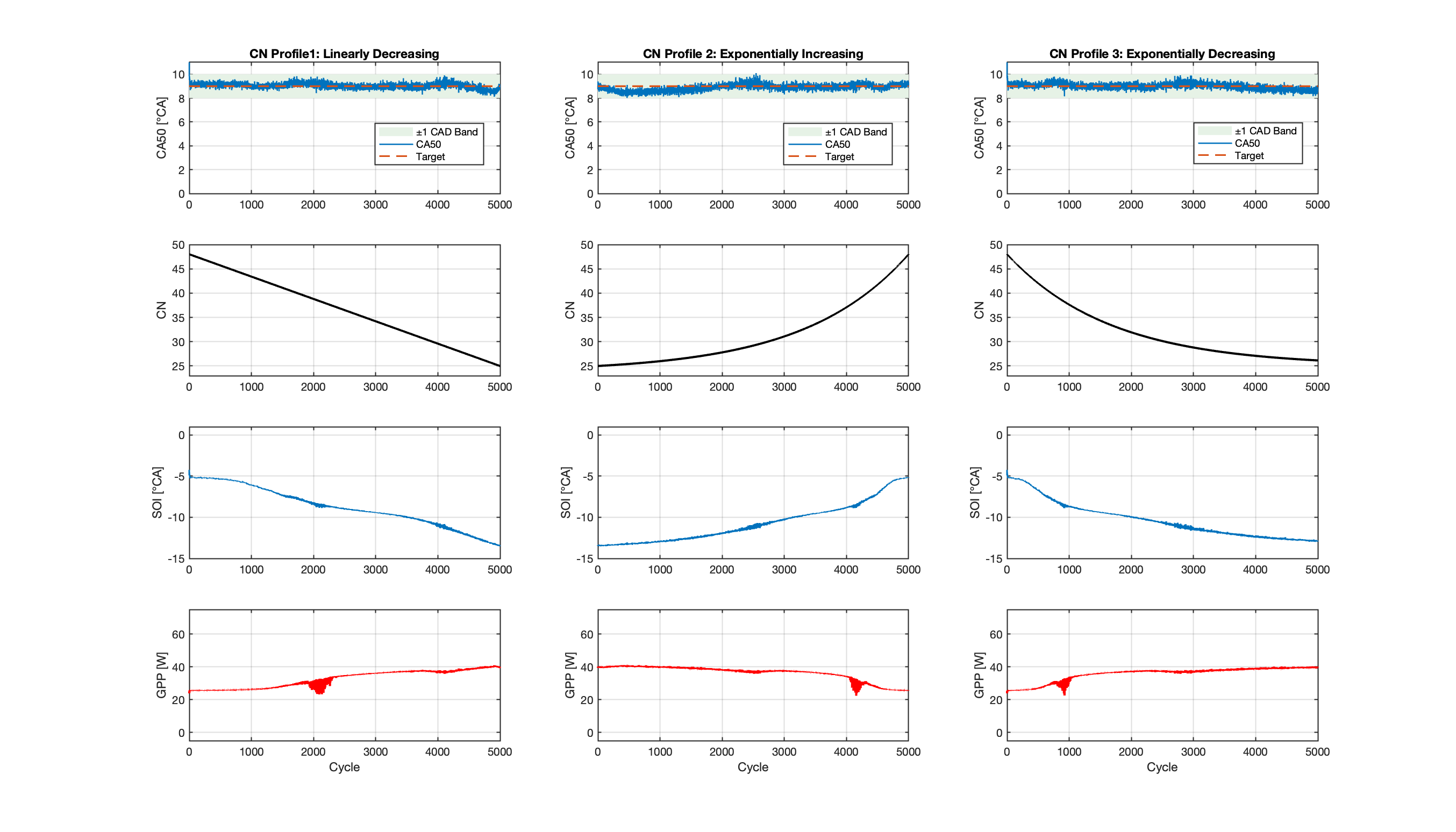}
\caption{Generalization of the CN-augmented controller on three unseen CN profiles. (i) MAE remains low at $0.187^\circ$, $0.235^\circ$, and $0.201^\circ$ CA for the linearly decreasing, exponentially increasing, and exponentially decreasing CN profiles, respectively. (ii) All errors remain within $1.3\times$ of the training MAE of $0.181^\circ$ CA. (iii) All profiles remain within the $\pm 1^\circ$ CA band. (iv) The actuator trajectories adapt smoothly in the expected physical direction with changing fuel reactivity.}
\label{fig:generalization}
\end{figure*}

The CN-augmented controller generalizes well to all three unseen CN profiles, preserving tight CA50 regulation and physically consistent actuator adaptation throughout the 5000-cycle horizon (Fig.~\ref{fig:generalization}). On both decreasing-CN profiles, namely linear CN $48\to25$ and exponential CN $48\to26$, SOI advances monotonically from approximately $-5.1^\circ$ CA at high CN to $-13.4^\circ$ CA at low CN, while GPP increases from approximately $25.5~\mathrm{W}$ to $40~\mathrm{W}$. On the exponentially increasing profile, CN $25\to48$, the opposite trends hold. SOI retards from $-13.4^\circ$ CA to $-5.2^\circ$ CA, and GPP decreases from $39.8~\mathrm{W}$ to $25.6~\mathrm{W}$. These directional responses are physically consistent with a fuel-reactivity-conditioned control law, confirming that the GRU-based CN estimate encodes the dominant latent variable with sufficient fidelity to direct the actuators correctly on dynamics and rate profiles absent from training.

Across all profiles, the controller maintains close regulation within the desired $1^{\circ}$ CA band around CA50 reference. These results indicate that the GRU-based CN estimate generalizes across unseen CN dynamics. Importantly, the actuator trajectories remain physically consistent across all these profiles. As fuel reactivity decreases, the controller advances SOI and increases GPP to provide greater ignition assistance, while increasing CN leads to retarded SOI and reduced GPP. This confirms that the learned policy captures the underlying fuel-reactivity dependence of the combustion process rather than memorizing a trajectory-specific actuation pattern.

The quantitative results in Table~\ref{tab:generalization} support the same conclusion. At the training reference, $y^{\star}=9^\circ$\,CA, all three unseen CN profiles show consistently accurate tracking. The exponentially increasing profile is relatively more challenging despite sharing the same CN direction as the training ramp, from $25$ to $48$, because its accelerating rate of change induces faster fuel-reactivity transitions than the constant-rate profile used during training. More importantly, the same controller, trained only at $y^{\star}=9^\circ$\,CA, maintains accurate tracking across all tested CA50 references without retraining. The results at $y^{\star}=8^\circ$\,CA and $y^{\star}=10^\circ$\,CA show that the policy generalizes well to modest shifts around the training reference. At $y^{\star}=7^\circ$\,CA, the tracking degradation is mainly associated with the exponentially increasing CN profile, while both decreasing profiles remain comparatively well regulated. The most visible degradation occurs at $y^{\star}=11^\circ$\,CA, and it appears asymmetrically across the CN profiles. The two decreasing-CN profiles are more affected than the increasing-CN profile. This asymmetry has a physical interpretation. Enforcing late combustion phasing at $11^\circ$\,CA is most difficult when CN simultaneously approaches its lowest values, which occurs near the end of the decreasing profiles as CN approaches $25$. In contrast, the exponentially increasing profile ends at high reactivity, with CN nearing $48$, where late phasing is easier to maintain. Across all tested setpoints and CN profiles, the controller retains the same physically interpretable actuator behavior observed at the training reference. The consistent generalization across these profiles, despite the GRU operating on CN trajectories absent from its training distribution, reflects the train-deploy consistency of the architecture.

\begin{table}[htpb!]
\centering
\caption{Tracking performance of the CN-augmented controller on unseen CN profiles across five CA50 reference values. The controller was trained at $y^{\star} = 9$$^\circ$ CA only.}
\label{tab:generalization}
\footnotesize
\setlength{\tabcolsep}{3pt}
\begin{tabular}{@{}llccc@{}}
\toprule
\textbf{CA50$^{\text{ref}}$} & \textbf{CN profile} & \textbf{MAE} & \textbf{RMSE} & \textbf{Avg MAE} \\
 & & [$^\circ$ CA] & [$^\circ$ CA] & [$^\circ$ CA] \\
\midrule
\multirow{3}{*}{7~$^\circ$ CA} & Linear dec. & 0.253 & 0.398 & \multirow{3}{*}{0.318} \\
 & Exp. increase & 0.458 & 0.699 & \\
 & Exp. decrease & 0.241 & 0.317 & \\
\midrule
\multirow{3}{*}{8$^\circ$ CA} & Linear dec. & 0.191 & 0.261 & \multirow{3}{*}{0.231} \\
 & Exp. increase & 0.272 & 0.364 & \\
 & Exp. decrease & 0.229 & 0.295 & \\
\midrule
\multirow{3}{*}{\textbf{9$^\circ$ CA}$^*$} & Linear dec. & 0.186 & 0.239 & \multirow{3}{*}{0.208} \\
 & Exp. increase & 0.235 & 0.299 & \\
 & Exp. decrease & 0.201 & 0.255 & \\
\midrule
\multirow{3}{*}{10$^\circ$ CA} & Linear dec. & 0.184 & 0.234 & \multirow{3}{*}{0.199} \\
 & Exp. increase & 0.225 & 0.288 & \\
 & Exp. decrease & 0.189 & 0.240 & \\
\midrule
\multirow{3}{*}{11$^\circ$ CA} & Linear dec. & 0.713 & 1.332 & \multirow{3}{*}{0.438} \\
 & Exp. increase & 0.221 & 0.291 & \\
 & Exp. decrease & 0.379 & 0.819 & \\
\bottomrule
\multicolumn{5}{@{}l}{\scriptsize $^*$Training reference value.}
\end{tabular}
\end{table}

\subsubsection{Validation on a realistic multi-rate CN stress-test trajectory}
\label{subsec:external_validation}

As a final validation beyond the canonical unseen ramps, the CN-augmented controller was tested on a 5000-cycle multi-rate CN trajectory emulating mixed-rate fuel-reactivity variation. This experiment uses the best training checkpoint and is therefore interpreted as validation of the converged policy. The controller maintains tight combustion-phasing regulation throughout (Fig.~\ref{fig:external_cn_validation}), achieving a CA50 MAE of $0.170^\circ$ CA, RMSE of $0.218^\circ$ CA, and $99.98\%$ of cycles within the $\pm 1^\circ$ CA band as CN decrease from 48 to 25.  The maximum absolute error of 2.080◦CA occurs at cycle 1 and is attributed to the empty-history initialization of the GRU observer and controller state before the estimator buffer has accumulated representative combustion history. The actuator trends remain physically consistent throughout. As CN decreases from $48$ to $25$, SOI advances continuously from approximately $-5.1^\circ$ CA to $-13.5^\circ$ CA, and GPP rises from approximately $26~\mathrm{W}$ to $40~\mathrm{W}$, correctly reflecting the monotonically increasing ignition-assist demand at lower fuel reactivity. Because this trajectory combines a rapid staircase descent with prolonged constant-CN dwells and abrupt inter-level drops, it constitutes a stronger deployability stress test than the smooth benchmark ramps alone.

\begin{figure}[htpb!]
\centering
\includegraphics[width=\columnwidth,height = 1.4\columnwidth]{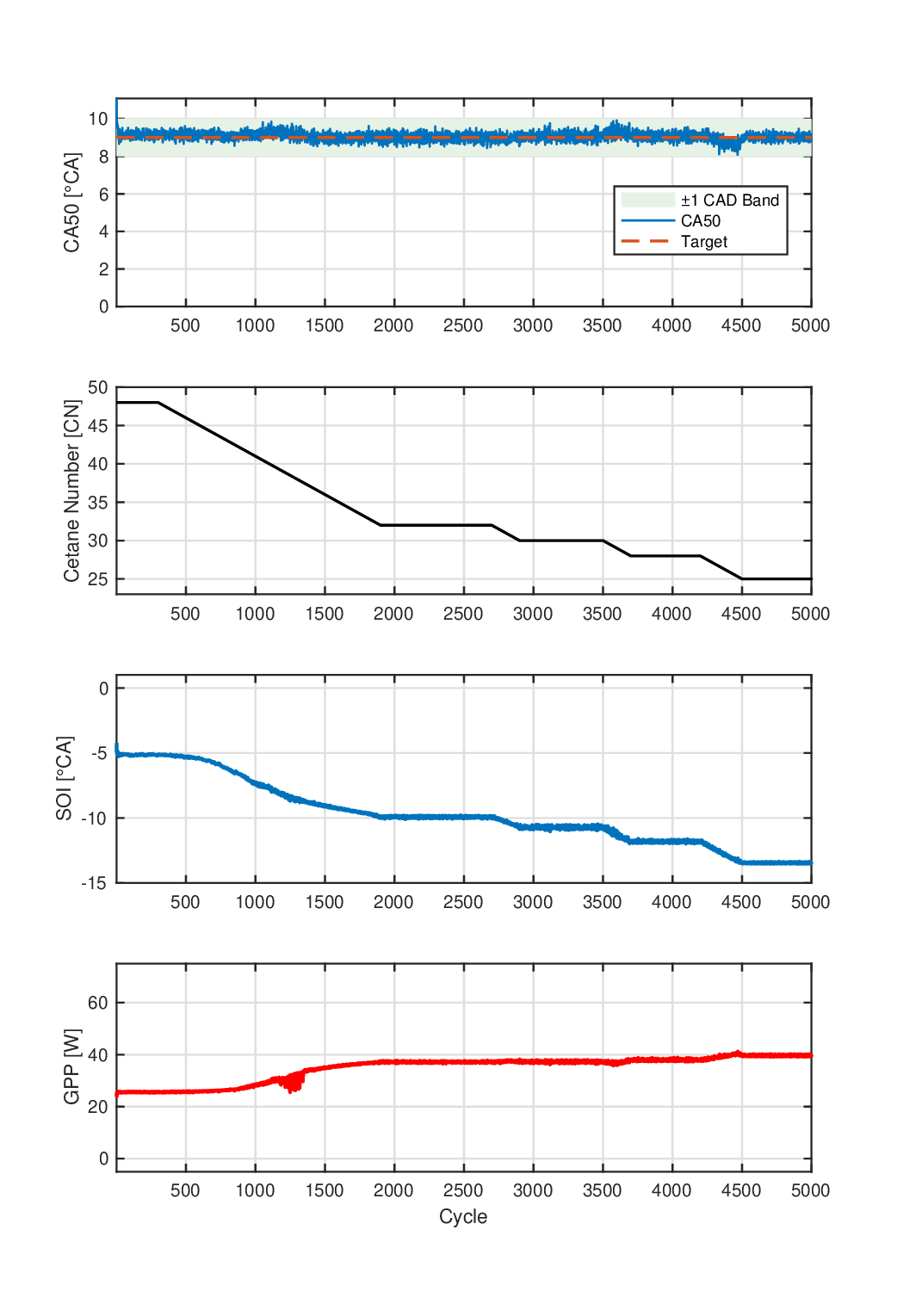}
\caption{Additional CN-augmented controller validation on a 5000-cycle multi-rate CN stress-test trajectory with MAE $=0.170^\circ$ CA, RMSE $=0.218^\circ$ CA, and $99.98\%$ of cycles within $\pm 1^\circ$ CA. As CN decreases from $48$ to $25$, SOI advances from $\sim -5.1^\circ$ CA to $\sim -13.5^\circ$ CA and GPP rises from $\sim 26$ to $\sim 40~\mathrm{W}$, demonstrating physically consistent reactivity-aware control under mixed-rate fuel-reactivity variation.}
\label{fig:external_cn_validation}
\end{figure}

Taken together, the results of Sections~\ref{subsec:ddpg_obs}-\ref{subsec:external_validation} show that controller evaluation in this multivariable combustion problem must assess both tracking accuracy and actuator realism. The observation-only baseline can preserve CA50 while converging to a physically incorrect high-assist shortcut, and the RDPG baseline shows that generic recurrent memory does not resolve this defect. Only the CN-augmented architecture achieves accurate tracking together with smooth, physically interpretable control, with the multi-rate stress test confirming that these properties persist beyond the smooth benchmark ramps. This is consistent with the POMDP analysis of Section~\ref{sec:problem}: by restoring approximate Markov structure through latent-state augmentation, the critic receives more consistent targets across CN regimes and the actor learns control actions aligned with combustion physics. The contribution is therefore not the addition of an auxiliary input per se, but the identification of a structural failure mode of conventional deep RL under hidden fuel reactivity and many-to-one combustion mappings, together with a separated estimation-control architecture that resolves it.

\begin{remark}
\label{subsec:estimator_quality}
The role of the deployable CN bottleneck is best understood at the controller level rather than as a standalone estimation objective. In the proposed architecture, $\hat{\theta}_k$ need not accurately reconstruct the true fuel property; it serves as a low-dimensional indicator of the prevailing fuel regime, enabling initialization in the appropriate operating region while CA50 error feedback provides fine-scale correction. Its value is structural: by exposing a compact representation of the dominant latent variable, it mitigates state aliasing within the control loop, yielding improved regulation and physically consistent actuation despite approximation error. Performance is therefore appropriately assessed through control quality and generalization rather than standalone estimation accuracy.
\end{remark}

\begin{remark}
\label{rem:lupi}
The proposed hidden-CN controller occupies a distinct position among methods that exploit training-time information asymmetry. In the LUPI framework \cite{vapnik2015learning} and asymmetric RL extensions \cite{baisero2021unbiased, ebi2025informed}, privileged latent variables are typically provided to the critic during training to improve gradient estimation. By contrast, the present architecture enforces deployment consistency: true CN is used only for offline training of the GRU estimator, while both actor and critic operate solely on $\hat{\theta}_k$ during policy learning and execution, eliminating any training-deployment representation gap. Although it forgoes the variance-reduction benefits of a privileged critic, it yields a practically aligned formulation for combustion control where CN is fundamentally unmeasured.
\end{remark}

\section{Conclusion}\label{sec:conclusion}

CA50 combustion phasing control in multi-fuel CI engines remains challenging under unknown and continuously varying fuel reactivity, since latent changes in cetane number (CN) perturb ignition delay and shift combustion dynamics from cycle to cycle. Although the combustion response is predominantly static with respect to control input, the unmeasured and temporally evolving CN makes the problem partially observable, so similar observable states may require different optimal control actions. Existing approaches do not fully address this information-structure challenge. Conventional methods remain calibration sensitive and adapt slowly under rapid CN variation. Data-driven estimate-then-control architectures improve performance, but their closed-loop behavior remains sensitive to real-time estimation error and latency. Moreover, when the controller is calibrated assuming accurate CN estimates rather than the noisy and lagged signal available at deployment, residual estimation errors propagate directly into control actions. This paper investigates learning formulations with increasing temporal and representational capacity, including LinCB, history-augmented CB, observation-only DDPG, recurrent DDPG (RDPG), and the proposed GRU-RL framework. The results show that CB degrades due to violation of i.i.d.\ context assumptions, fixed-history augmentation provides only partial improvement, observation-only RL remains limited by latent-state ambiguity, and generic recurrence alone is insufficient under rapid CN variation. The proposed GRU-RL framework addresses these limitations by integrating latent fuel-state inference with sequential decision-making. A GRU is trained offline to produce a deployable fuel-reactivity estimate from combustion history, and both actor and critic are optimized exclusively under this noisy estimate rather than oracle CN. This makes the learned policy implicitly robust to estimation approximation errors and removes the information gap between training and deployment. The resulting controller achieves reliable CA50 regulation even across unseen CN trajectories, with mean absolute tracking error below $0.25^\circ\,\mathrm{CA}$ at the training setpoint, smooth actuation. The generalization results further show that the GRU estimate remains useful on CN trajectories outside its training distribution and that the policy preserves physically consistent actuator behavior. Overall, the study establishes that supervised recovery of fuel-state information from history, when coupled with sequential policy optimization, is more effective than myopic contextual decision-making, observation-only learning, or generic recurrence alone. It also highlights train-deploy consistency, achieved by optimizing the policy under the deployed estimation signal rather than privileged oracle information, as a critical design principle for learning-based combustion control. Future work will extend the framework to higher-dimensional operating conditions and validate it on real engine hardware.

\section*{Declaration of conflicting interests}
The authors declared no potential conflicts of interest with respect to the research, authorship, and/or publication of this article.

\section*{Funding}
Research was sponsored by the DEVCOM Army Research Laboratory and was accomplished under Cooperative Agreement Number W911NF-20-2-0161. The views and conclusions contained in this document are those of the authors and should not be interpreted as representing the official policies, either expressed or implied, of the DEVCOM Army Research Laboratory of the U.S. Government. The U.S. Government is authorized to reproduce and distribute reprints for Government.

\section*{Acknowledgements}
The authors would like to thank the team at Engine Research center of University of Wisconsin Madison for providing us with the experimental data for this work. The authors would also like to mention that this work used computational resources at the Minnesota Supercomputing Institute (MSI), University of Minnesota.

\bibliographystyle{elsarticle-num}
\bibliography{Ref_BIO2}
\end{document}